\documentclass[onecolumn,accepted=2017-07-07]{quantumarticle}
\usepackage[english]{babel}
\usepackage[T1]{fontenc}
\usepackage[utf8]{inputenc}
\usepackage{textcomp}
\usepackage{amsmath}
\usepackage{amsthm}
\usepackage{amstext}
\usepackage{amssymb}
\usepackage{graphicx}
\usepackage[unicode=true,pdfusetitle,
 bookmarks=true,bookmarksnumbered=false,bookmarksopen=false,
 breaklinks=true,pdfborder={0 0 0},backref=false,colorlinks=true]
 {hyperref}
\usepackage[numbers,sort&compress]{natbib}

\makeatletter

\usepackage{bbm}

\newcommand\eqdef{\mathrel{\overset{\makebox[0pt]{\mbox{\normalfont\tiny def}}}{=}}}

\usepackage{footmisc}

\usepackage{xcolor}

\makeatother

\providecommand{\lemmaname}{Lemma}
\newtheorem*{lem*}{\protect\lemmaname}

\begin{document}

\title{A Stronger Theorem Against Macro-realism}	

\date{\today}

\author{John-Mark A. Allen}
\affiliation{Department of Computer Science, University of Oxford, Wolfson Building, Parks Road, Oxford, OX1 3QD, United Kingdom}
\email{john-mark.allen@cs.ox.ac.uk}
\thanks{Corresponding author}
\author{Owen J. E. Maroney}
\affiliation{Faculty of Philosophy, University of Oxford, Radcliffe Humanities, Woodstock Road, Oxford, OX2 6GG, United Kingdom}
\author{Stefano Gogioso}
\address{Department of Computer Science, University of Oxford, Wolfson Building, Parks Road, Oxford, OX1 3QD, United Kingdom}

\begin{abstract}
Macro-realism is the position that certain ``macroscopic'' observables
must always possess definite values: \emph{e.g.} the table is in some
definite position, even if we don't know what that is precisely. The
traditional understanding is that by assuming macro-realism one can
derive the Leggett-Garg inequalities, which constrain the possible
statistics from certain experiments. Since quantum experiments can
violate the Leggett-Garg inequalities, this is taken to rule out the
possibility of macro-realism in a quantum universe. However, recent
analyses have exposed loopholes in the Leggett-Garg argument, which
allow many types of macro-realism to be compatible with quantum theory
and hence violation of the Leggett-Garg inequalities. This paper takes
a different approach to ruling out macro-realism and the result is
a no-go theorem for macro-realism in quantum theory that is stronger
than the Leggett-Garg argument. This approach uses the framework of
ontological models: an elegant way to reason about foundational issues
in quantum theory which has successfully produced many other recent
results, such as the PBR theorem.
\end{abstract}
\maketitle

\section{Introduction\label{sec:Introduction}}

The concept of macro-realism was introduced to the study of quantum
theory by Leggett \& Garg alongside their eponymous inequalities \cite{Leggett85}.
Loosely, macro-realism is the philosophical position that certain
``macroscopic'' quantities always possess definite values. The Leggett-Garg
inequalities (LGIs) are inequalities on observed measurement statistics
that are derived by assuming a particular form of macro-realism and
can be violated by measurements on quantum systems. The purpose of
the LGIs is therefore to prove that quantum theory and macro-realism
are incompatible. However, since its introduction the exact meaning
of ``macro-realism'' has been the subject of debate \cite{Ballentine87,Leggett87,Leggett88,Leggett02a,Leggett02b,KoflerBrukner13,Maroney14}. 

Recently, there has been a surge of interest in violation of the LGIs
from both physical and philosophical angles. The review in Ref. \cite{Emary14}
comprehensively covers experimental and theoretical work up to 2014.
More recent experimental work has focussed on noise tolerance and
closing experimental loopholes \cite{Knee+16,HuffmanMizel16,Zhou+15}.
Also, several theoretical investigations have aimed to interrogate
and clarify exactly what is required to derive the LGIs and what is
implied by their experimental violation \cite{Maroney14,ClementeKofler15,ClementeKofler16,KoflerBrukner13,Bacciagaluppi15,WildeMizel12}.
This paper follows the clarifying work of Ref.~\cite{Maroney14}.

Macro-realism is an ontological position; that is, the statement that
a certain ``macroscopic'' quantity is macro-realist is a statement
about the real state of affairs, the \emph{ontology}, of the universe.
In the field of quantum foundations, the framework of \emph{ontological
models} has been developed as a way to analyse such statements in
generality, making as few assumptions as possible. To use an ontological
model to describe a system requires just two core assumptions. First,
that the system being described has some \emph{ontological state}---the
real fact about how the system actually is (of course, which ontological
state is currently occupied is generally unknown). Second, that standard
probability theory can be applied to the ontological states. Their
generality has made ontological models very useful and they have been
used to derive and clarify many important results in quantum foundations
including: Bell's theorem \cite{Bell87}, the PBR theorem \cite{Pusey12},
and excess baggage \cite{Hardy04}. Reference~\cite{Leifer14b} comprehensively
reviews many of these results.

By using ontological models it is possible to illuminate and classify
various definitions of macro-realism precisely \cite{Maroney14}.
This analysis reveals some fundamental loopholes in the Leggett-Garg
argument for the incompatibility of quantum theory and macro-realism.
In particular, it shows that violation of the LGIs serves only to
rule out one subset of macro-realist models and that there are other
macro-realist models of quantum theory which are compatible with the
LGIs. For example, Bohmian quantum theory \cite{Bohm95,Durr09,Bohm52a,Bohm52b,deBroglie27}
can be viewed as a macro-realist model which reproduces all predictions
of quantum theory and therefore cannot be ruled out by the Leggett-Garg
argument. These loopholes are not experimental but logical; the only
way to close them is to fundamentally change the argument.

In this paper, a stronger theorem for the incompatibility between
quantum theory and macro-realism is presented. This theorem closes
a loophole in the Leggett-Garg argument and establishes that quantum
theory is incompatible with a larger subset of macro-realist models.
It does not prove incompatibility of quantum theory with all macro-realist
models since this is not possible; such a result would be in conflict
with the existence of the theory of Bohmian mechanics (section~\ref{sec:Loopholes-in-the}).
The theorem proceeds in a very different manner than the Leggett-Garg
argument and is related to the main theorem from Ref.~\cite{Allen16}.
It thereby circumvents many of the controversies of the original Leggett-Garg
approach.

It should be noted that mathematically there is no meaning to the
stipulation that macro-realism is about ``macroscopic'' quantities,
as opposed to other physical quantities that aren't ``macroscopic''.
Philosophically, however, it is easy to understand the desire for
macro-realism applying to ``macroscopic'' quantities. The types
of physical quantity that humans experience are all considered macroscopic
and they certainly appear to possess definite values. On the other
hand, it is much easier to imagine that microscopic quantities that
aren't directly observed behave in radically different ways. So while
there is nothing in the structure of quantum theory to pick-out ``macroscopic''
versus ``microscopic'', the motivation for considering macro-realism
does come from considering macroscopic quantities, hence the name.

This paper is organised as follows. The framework of ontological models
is introduced in section~\ref{sec:Ontological-Models}. This is then
used in section~\ref{sec:Macro-realism} to give the precise definitions
of macro-realism and its three sub-sets. Both of these sections are
without reference to quantum theory and are entirely independent of
it. This is as it should be, since those concepts apply to descriptions
of any physical system based on philosophical assumptions and do not
depend on any specific physical theory (of course their most common
applications are with quantum systems). In section~\ref{sec:Macro-realism-in-Quantum}
macro-realism is applied to quantum theory and several useful definitions
and lemmas about quantum ontological models are presented. Section~\ref{sec:Loopholes-in-the}
discusses the Leggett-Garg argument against macro-realism and explains
why the existence of loopholes \cite{Maroney14} restrict it to only
a simple class of macro-realist models. Section~\ref{sec:A-Stronger-No-Go}
then proves the main theorem which introduces a new approach, ruling
out a larger class of macro-realist models than Leggett-Garg while
using weaker assumptions. A discussion follows in section~\ref{sec:Discussion}
where conclusions are drawn, together with discussion of further research
directions including the possibility for experiments based on this
result.

\section{Ontological Models\label{sec:Ontological-Models}}

When we debate types of ``realism'' in quantum theory, we are normally
making an ontological argument. We're trying to say something about
the underlying actual state of affairs: whether such a thing exists,
how it can or can't behave, and so on. So it is with macro-realism.
The macro-realist, loosely, believes that the underlying ontology
in some definite sense \emph{possesses} a value for certain macroscopic
quantities at all times. The framework of ontological models \cite{Harrigan07,Harrigan10}
has been developed to make discussions about ontology in physics precise
and is the natural arena for such discussions.

An ontological model is exactly that: a bare-bones model for the underlying
ontology of some physical system. The system may also be correctly
described by some other, higher, theory---such as Newtonian mechanics
or a quantum theory---in which case the ontological model must be
constrained to reproduce the predictions of that theory. By combining
these constraints with the very general framework of ontological models,
interesting and general conclusions can be drawn about the nature
of the ontology. It is important to note that, while ontological models
are normally used to discuss quantum ontology, the framework itself
is entirely independent from quantum theory. The presentation of ontological
models here follows Ref.~\cite{Leifer14b}, which contains a much
more thorough discussion.

As noted above, the framework of ontological models relies on just
two core assumptions: 1) that the system of interest has some ontological
state $\lambda$ and 2) that standard probability theory may be applied
to these states. Together, these bring us to consider the ontology
of some physical system as represented by some measurable set $\Lambda$
of ontic states $\lambda\in\Lambda$ which the system might occupy.
The requirement that $\Lambda$ be measurable guarantees that probabilities
over $\Lambda$ can be defined.

In the lab, a system can be prepared, transformed, and measured in
certain ways. Each of these operational processes needs to be describable
in the ontological model.

Preparation must result in the system ending up in some ontic state
$\lambda$, though the exact state need not be known. Thus, each preparation
$P$ gives rise to some \emph{preparation measure }$\mu$ over $\Lambda$
which is a probability measure ($\mu(\emptyset)=0$, $\mu(\Lambda)=1$).
For every measurable subset $\Omega\subseteq\Lambda$, $\mu(\Omega)$
gives the probability that the resulting $\lambda$ is in $\Omega$.

Similarly, a transformation $T$ of the system will generally change
the ontic state from $\lambda^{\prime}\in\Lambda$ to a new $\lambda\in\Lambda$.
Recalling that the ontic state $\lambda^{\prime}$ represents the
entirety of the actual state of affairs before the transformation,
then the final state can only depend on $\lambda^{\prime}$ (and not
the preparation method or any previous ontic states, except as mediated
through $\lambda^{\prime}$). The transformations must therefore be
described as \emph{stochastic maps} $\gamma$ on $\Lambda$. A stochastic
map consists of a probability measure $\gamma(\cdot|\lambda^{\prime})$
for each initial ontic state, such that for any measurable $\Omega\subseteq\Lambda$,
$\gamma(\Omega|\lambda^{\prime})$ is the probability that the final
$\lambda$ lies in $\Omega$ given that the initial state was $\lambda^{\prime}$
\footnote{These stochastic maps, viewed as functions $\gamma(\Omega|\cdot)\,:\,\Lambda\rightarrow[0,1]$
(one for each measurable $\Omega\subseteq\Lambda$), must be \emph{measurable
functions}. That is, for any measurable set $\mathcal{S}\subseteq[0,1]$
and any $\gamma(\Omega|\lambda^{\prime})$, then $\{\lambda^{\prime}\in\Lambda\,:\,\gamma(\Omega|\lambda^{\prime})\in\mathcal{S}\}\subseteq\Lambda$
is a measurable set.}.

Finally, a measurement $M$ may give rise to some outcome $E$. Again,
which outcome is obtained can only depend on the current ontic state
$\lambda^{\prime}$. Therefore a measurement $M$ gives rise to a
conditional probability distribution $\mathbb{P}(E|\lambda^{\prime})$
\footnote{These probability distributions, viewed as functions $\Lambda\rightarrow[0,1]$,
must also be measurable functions.}. For this paper it is only necessary to consider measurements that
have countable sets of possible outcomes $E$.

Putting these parts together: if we have a system where a preparation
$P$ is performed followed by some transformation $T$ and some measurement
$M$ then the ontological model for that system must have some preparation
measure $\mu$, stochastic map $\gamma$, and conditional probability
distribution $\mathbb{P}$ such that the probability of obtaining
outcome $E$ is
\begin{equation}
\int_{\Lambda}\mathrm{d}\nu(\lambda)\,\mathbb{P}(E\,|\,\lambda)\label{eq:Ontological-model-probability}
\end{equation}
where 
\begin{equation}
\nu(\Omega)\eqdef\int_{\Lambda}\mathrm{d}\mu(\lambda)\,\gamma(\Omega\,|\,\lambda)\label{eq:Ontological-model-transformation}
\end{equation}
is the effective preparation measure obtained by preparation $P$
followed by transformation $T$. 

Note that ontological models are required to be closed under transformations.
That is, for any preparation $\mu$ and transformation $\gamma$ in
the model then the preparation $\nu$ defined by Eq.~(\ref{eq:Ontological-model-transformation})
must also exist in the model (since a preparation followed by a transformation
is itself a type of preparation).

So an ontological model for some physical system does the following:
\begin{enumerate}
\item defines a measurable set $\Lambda$ of ontic states for the system;
\item for each possible transformation method $T$ defines a stochastic
map $\gamma$ from $\Lambda$ to itself;
\item for each possible preparation method $P$ defines a preparation measure
$\mu$ over $\Lambda$, ensuring closure under the actions of the
stochastic maps as in Eq.~(\ref{eq:Ontological-model-transformation});
\item for each possible measurement method $M$ defines a conditional probability
distribution $\mathbb{P}$ over the outcomes given $\lambda\in\Lambda$;
\end{enumerate}
and then produces probabilities for measurement outcomes via Eqs.~(\ref{eq:Ontological-model-probability},\ref{eq:Ontological-model-transformation}).
The possible ontological models for a system can be constrained by
requiring that these probabilities match those given by other theories
known to accurately describe it (or probabilities obtained by experiment).

It should be noted that ontological models are not usually defined
in the measure-theoretic way presented here, but are often presented
using probability \emph{distributions} rather than measures. However
it has been noted in Ref.~\cite{Leifer14b} that this simplification
precludes many reasonably ontological models, including the archetypal
Beltrametti-Bugajski model \cite{Beltrametti95} which simply takes
ontic states to be quantum states. In order to do justice to macro-realism
the more accurate approach has therefore been taken here.

\section{Macro-realism\label{sec:Macro-realism}}

Exactly what is meant by ``macro-realism'' has been a subject of
contention ever since its introduction alongside the LGIs. This controversy
has fed into more recent work on understanding the violation of the
LGIs \cite{ClementeKofler15,ClementeKofler16,Bacciagaluppi15}. In
Ref.~\cite{Maroney14}, uses of the term ``macro-realism'' are
analysed and the concept is illuminated using ontological models.
One result of that paper is that the ``macro-realism'' intended
by Leggett and Garg, as well as many subsequent authors, can be made
precise in a reasonable way with the definition:

\emph{``A macroscopically observable property with two or more distinguishable
values available to it will at all times determinately possess one
or other of those values.'' \cite{Maroney14}}

Throughout this paper, ``distinguishable'' will be taken to mean ``in principle perfectly distinguishable by a single measurement in the noiseless case''. Note that macro-realism is defined with respect to some specific property $Q$. A macro-realist model might (and generally will) be macro-realist for some properties and not others. This property will have some values $\{q\}$ and to be ``observable'' must correspond to at least one measurement $M_{Q}$ with corresponding outcomes $E_{q}$.

Reference~\cite{Maroney14} fleshes out this definition using ontological
models and as a result describes three sub-categories of macro-realism.
In order to discuss these it will be necessary to first define an
\emph{operational eigenstate} in ontological models.

An operational eigenstate $Q_{q}$ of any value $q$ of an observable
property $Q$ is a set of preparation procedures $\{P_{q}\}$. This
set is defined so that immediately following any $P_{q}$ with any
measurement of the quantity $Q$ will result in the outcome $E_{q}$
with certainty. In other words, an operational eigenstate is simply
an extension of the concept of a quantum eigenstate to ontological
models: the preparations which, when appropriately measured, always
return a particular value of a particular property. Note that if two
values $q,q^{\prime}$ have operational eigenstates then they can
sensibly be called ``distinguishable'', since any system prepared
in a corresponding operational eigenstate can be identified to have
one value and not the other with certainty.

The three sub-categories of macro-realism for some quantity $Q$ are
then:
\begin{enumerate}
\item \emph{Operational eigenstate mixture macro-realism (EMMR) }--\emph{
}The only preparations in the model are operational eigenstates of
$Q$ or statistical mixtures of operational eigenstates. That is,
for each preparation $P_{q,i}$ of each operational eigenstate $Q_{q}$
let $\mu_{q,i}$ be the preparation measure and let $\{c_{q,i}\}_{q,i}$
be a set of positive numbers summing to unity. In EMMR every preparation
measure can be written in the form $\nu=\sum_{q}\sum_{i}c_{q,i}\mu_{q,i}$.
Note that this means that the space of ontic states $\Lambda$ need
only include those $\lambda$ accessible by preparing some operational
eigenstate of $Q$, as no other ontic states can ever be prepared.
\item \emph{Operational eigenstate support macro-realism (ESMR) }-- Like
EMMR, every ontic state $\lambda\in\Lambda$ accessible by preparing
some operational eigenstate but, unlike EMMR, there are preparation
measures in the model that are not statistical mixtures of operational
eigenstate preparations for $Q$. That is, let $\Omega$ be a measurable
subset for which $\mu_{q}(\Omega)=1$ for every operational eigenstate
preparation $\mu_{q}$. Then for every preparation measure $\nu$
in an ESMR model, $\nu(\Omega)=1$ for all such $\Omega$. Moreover,
there exists some preparation procedure not in the mixture form required
by EMMR. In other words, if you're certain to prepare an ontic state
from some subset $\Omega$ when preparing an operational eigenstate
of $Q$, then you're also certain to prepare an ontic state from $\Omega$
from any other preparation measure in the model.
\item \emph{Supra eigenstate support macro-realism (SSMR)} -- Every ontic
state $\lambda$ in the model will produce some specific value $q_{\lambda}$
of $Q$ when a measurement of $Q$ is made, but some of those ontic
states are not accessible by preparing any operational eigenstate
of $Q$. That is, for every $\lambda\in\Lambda$ then there is some
value $q_{\lambda}$ of $Q$ such that $\mathbb{P}_{M}(E_{q_{\lambda}}\,|\,\lambda)=1$
for every $\mathbb{P}_{M}$ corresponding to a measurement of $Q$.
Moreover, there exists some measurable subset $\Omega\subset\Lambda$
and preparation measure $\nu$ such that $\nu(\Omega)>0$ while $\mu_{q}(\Omega)=0$
for every preparation measure $\mu_{q}$ of an operational eigenstate
$Q_{q}$.
\end{enumerate}
To help unpack these definitions, they are illustrated in Fig.~\ref{fig:MR-figure}.

\begin{figure}
\includegraphics[width=1\columnwidth]{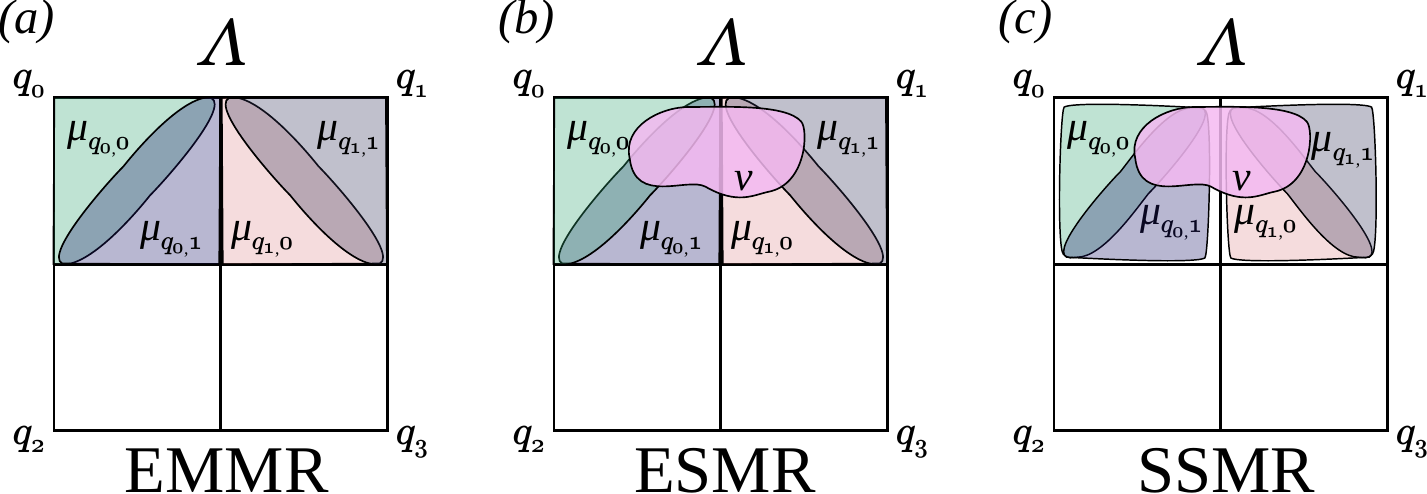}

\protect\caption{Illustration of the three sub-categories of macro-realism as defined
in the text. In each case the large square represents the whole ontic
state space $\Lambda$, the four smaller squares indicate those subspaces
of ontic states associated with each value $q_{0..3}$ of some quantity
$Q$, and the shaded regions represent those ontic states accessible
by preparing some select preparation measures.\protect \\
(a) illustrates EMMR, where the squares for each $q_{i}$ are all
ontic states preparable via some operational eigenstate preparation $\mu_{q_{i},j}$
and all other allowed preparation measures are simply statistical
mixtures of these, \emph{e.g. $\nu=\frac{1}{3}\left(\mu_{q_{0},0}+\mu_{q_{0},1}+\mu_{q_{1},0}\right)$
}is permissible.\protect \\
(b) illustrates ESMR, where the state space is exactly as in EMMR,
but now more general preparation measures, such as the $\nu$ illustrated,
are permitted.\protect \\
(c) illustrates SSMR, where now every $\lambda$ in the box for $q_{i}$
must produce outcome $q_{i}$ in any appropriate measurement of $Q$,
but the operational eigenstates no longer fill these boxes. That is,
there are ontic states that lie outside the preparations for operational
eigenstates. General preparation measures over the boxes are still
permitted.}
\label{fig:MR-figure}

\end{figure}

In each of these cases, every ontic state $\lambda$ (up to possible
measure-zero sets of exceptions\footnote{One persuasive advantage to defining ontological models with probability
distributions, rather than measures, is that it eliminates the need
to append ``up to possible measure-zero sets of exceptions'' to
many otherwise-clear statements. We apologise for the necessary repetition
of this phrase in this paper.}) is associated with a specific value $q_{\lambda}$ of $Q$, such
that it can be sensibly said that $\lambda$ ``possesses'' $q_{\lambda}$.
This is why they are all considered types of macro-realism. Let's
consider this for each case in turn.

In an EMMR model, any preparation can be viewed as a choice between
preparations of operational eigenstates for values of $Q$, so any
resulting ontic state $\lambda$ ``possesses'' the corresponding
value $q$ since it could have been obtained by preparing an operational
eigenstate of $q$.

In an ESMR model, for every ontic state $\lambda$ accessible by preparing
some non-operational-eigenstate measure $\nu$, $\lambda$ can also
be prepared by an operational eigenstate of exactly one value of $Q$
(up to measure-zero sets of exceptions) and so similarly each ontic
state ``possesses'' the corresponding value of $Q$. 

In SSMR models the link between each $\lambda\in\Lambda$ and the
corresponding $q_{\lambda}$ is explicit in the definition. For any
$\lambda$ the $q_{\lambda}$ is that value for which $\mathbb{P}_{M}(E_{q_{\lambda}}\,|\,\lambda)=1$
as specified in the above definition. Thus, $\lambda$ ``possesses''
the value which it must return with certainty in any relevant measurement.

Note that these three sub-categories of macro-realism are defined
such that they are mutually exclusive, but they still have a natural
hierarchy to them. EMMR can be seen as a more restrictive variation
on ESMR, since you can make an EMMR model into an ESMR model simply
by including a single preparation measure that is not a statistical
mixture of operational eigenstate preparations (the ontic state space
and everything else can remain unchanged). Similarly, SSMR can be
seen as a less restrictive variation on ESMR. In ESMR, every ontic
state $\lambda$ can be obtained by preparing an operational eigenstate
preparation for a value of $Q$, by definition of operational eigenstate
it follows that a measurement of $Q$ will therefore return some specific
value for each ontic state (up to measure-zero sets of exceptions),
which is the requirement on the ontic states for SSMR.

\section{Macro-realism in Quantum Theory\label{sec:Macro-realism-in-Quantum}}

Now that macro-realism and ontological models have been defined independently
from quantum theory, it is time to apply these definitions to quantum
systems. Doing this will enable precise discussion of the loopholes
in the Leggett-Garg argument and proof of a stronger theorem. It will
also be necessary to state some useful definitions and lemmas from
ontological models.

\subsection{Ontological Models for Quantum Systems}

The postulates of quantum theory assign some Hilbert space $\mathcal{H}$
to a quantum system with dimension $d$ so that the set of physical
pure quantum states is $\mathcal{P}(\mathcal{H})=\{|\psi\rangle\in\mathcal{H}\,:\,\left\Vert \psi\right\Vert =1,\,|\psi\rangle\sim\mathrm{e}^{i\theta}|\psi\rangle\}$.
They also assign unitary operators on $\mathcal{H}$ to transformations
and orthonormal bases over $\mathcal{H}$ to measurements. For simplicity,
consider only systems with $d<\infty$. With this in mind, ontological
models for quantum systems can be described in generality.

An ontological model for a quantum system is defined by some ontic
state space $\Lambda$ as well as the relevant preparation measures,
stochastic maps, and conditional probability distributions. For each
state $|\psi\rangle\in\mathcal{P}(\mathcal{H})$ there must be a set
$\Delta_{|\psi\rangle}$ of preparation measures $\mu_{|\psi\rangle}$---potentially
one for each distinct\emph{ }method for preparing $|\psi\rangle$.
Similarly, for each unitary operator $U$ on $\mathcal{H}$ there
is a set $\Gamma_{U}$ of stochastic maps $\gamma_{U}$ and for each
basis measurement $M=\{|i\rangle\}_{i=0}^{d}$ there is a set $\Xi_{M}$
of conditional probability distributions $\mathbb{P}_{M}$---again,
potentially one stochastic map/probability distribution for each experimental
method for transforming/measuring.

In order to investigate the properties of possible ontological models
it is required that the ontological model is capable of reproducing
the predictions of quantum theory. That is, for any $|\psi\rangle\in\mathcal{P}(\mathcal{H})$,
$\mu\in\Delta_{|\psi\rangle}$, $U,$ $\gamma\in\Gamma_{U}$, basis
$M$, and $\mathbb{P}_{M}\in\Xi_{M}$ it is required that
\begin{equation}
|\langle i|U|\psi\rangle|^{2}=\int_{\Lambda}\mathrm{d}\nu(\lambda)\,\mathbb{P}_{M}(|i\rangle\,|\,\lambda)\quad\forall|i\rangle\in M\label{eq:ontological-model-quantum-probability}
\end{equation}
where $\nu$ is defined as in Eq.~(\ref{eq:Ontological-model-transformation}).
Note also that $\nu\in\Delta_{U|\psi\rangle}$ since preparing the
quantum state $|\psi\rangle$ (via any ontological preparation $\mu\in\Delta_{|\psi\rangle}$)
followed by performing the quantum transformation $U$ (via any $\gamma\in\Gamma_{U}$)
is simply a way to prepare the quantum state $U|\psi\rangle$.

\subsection{State Overlaps}

In order to properly discuss macro-realism in quantum theory it is
necessary to discuss how to quantify state overlaps in ontological
models. In quantum theory any pair of non-orthogonal states $|\psi\rangle,|\phi\rangle\in\mathcal{P}(\mathcal{H})$
overlap by an amount quantified by the Born rule probability $|\langle\psi|\phi\rangle|^{2}$.
That is, for a system prepared in state $|\psi\rangle$ the probability
for it to behave (for all intents and purposes) like it was prepared
in state $|\phi\rangle$ is $|\langle\psi|\phi\rangle|^{2}$.

Adapting this logic to an ontological model for the quantum system,
consider the probability that a system prepared according to measure
$\mu$ behaves like it was prepared according to $\nu$. That is,
the probability that the ontic state obtained from $\mu$ could also
have been obtained from $\nu$. This quantity is called the \emph{asymmetric
overlap} and is mathematically defined as \cite{Allen16,Ballentine14,Leifer13b,Maroney12a}
\begin{equation}
\varpi(\nu\,|\,\mu)\eqdef\inf\{\mu(\Omega)\,:\,\Omega\subseteq\Lambda,\,\nu(\Omega)=1\},
\end{equation}
recalling that the infimum of a subset of real numbers is the greatest lower bound of that set. This is because a preparation of $\nu$ has unit probability of producing a $\lambda$ from each measurable subset $\Omega\subseteq\Lambda$ that satisfies $\nu(\Omega)=1$ \footnote{When defining $\varpi$ in this paper, the subsets $\Omega\subseteq\Lambda$ that are extremised over are all taken to be \emph{measurable} subsets of $\Lambda$. This detail is omitted in the text for clarity.}. Therefore by taking the minimum such $\Omega$, $\mu(\Omega)$ gives the desired probability.

It is not difficult to see that---for an ontological model of a quantum system---the Born rule upper bounds the asymmetric overlap. Almost all\footnote{Here, as elsewhere in measure theory, ``almost'' is used to mean ``up to measure-zero sets of exceptions''.} ontic states that can be obtained by preparing $|\phi\rangle$ will also return outcome $|\phi\rangle$ in any relevant measurement and so if a preparation of $|\psi\rangle$ results in a $\lambda$ that could have been obtained by preparing $|\phi\rangle$ then we know that this ontic state will almost surely return $|\phi\rangle$ in a relevant measurement. It follows that for any $\mu\in\Delta_{|\psi\rangle}$ and $\nu\in\Delta_{|\phi\rangle}$
\begin{equation}
\varpi(\nu\,|\,\mu)\leq|\langle\phi|\psi\rangle|^{2}.\label{eq:overlap-bound}
\end{equation}
A full proof of this is provided in Appendix~\ref{sec:app:Bounding-Asymmetric-Overlaps}.

It is useful to overload the definition of asymmetric overlap to include
the probability that preparing an ontic state via $\mu$ will produce
a $\lambda$ accessible by preparing some \emph{quantum state} $|\phi\rangle$.
This corresponds to 
\begin{equation}
\varpi(|\phi\rangle\,|\,\mu)\eqdef\inf\left\{ \mu(\Omega)\;:\;\nu(\Omega)=1,\;\forall\nu\in\Delta_{|\phi\rangle}\right\} ,
\end{equation}
which is clearly also upper bounded by the Born rule $\varpi(|\phi\rangle\,|\,\mu)\leq|\langle\phi|\psi\rangle|^{2}$.

The next useful generalisation is the overlap of some preparation
measure $\mu$ with \emph{two} quantum states $|0\rangle,|\phi\rangle$.
This can be thought of as the \emph{union} of the overlaps expressed
by $\varpi(|\phi\rangle\,|\,\mu)$ and $\varpi(|0\rangle\,|\,\mu)$
and is mathematically defined as
\begin{equation}
\varpi(|0\rangle,|\phi\rangle\,|\,\mu)\eqdef\inf\left\{ \mu(\Omega)\;:\;\nu(\Omega)=\chi(\Omega)=1,\;\forall\nu\in\Delta_{|\phi\rangle},\chi\in\Delta_{|0\rangle}\right\}
\end{equation}
or, equivalently as
\begin{equation}
\varpi(|0\rangle,|\phi\rangle\,|\,\mu) = \inf\left\{ \mu(\Omega_\phi \cup \Omega_0)\;:\;\nu(\Omega_\phi)=\chi(\Omega_0 )=1,\;\forall\nu\in\Delta_{|\phi\rangle},\chi\in\Delta_{|0\rangle}\right\}.
\end{equation}
So in this way $\varpi(|0\rangle,|\phi\rangle\,|\,\mu)$ expresses
the probability that sampling from $\mu$ produces a $\lambda$ accessible
by preparing either $|0\rangle$ or $|\phi\rangle$.

Since $\varpi(|0\rangle,|\phi\rangle\,|\,\mu)$ expresses the probability
of a disjunction of two events that have probabilities $\varpi(|0\rangle\,|\,\mu)$
and $\varpi(|\phi\rangle\,|\,\mu)$, it follows (by Boole's inequality) that it is bounded as follows
\begin{equation}
\varpi(|0\rangle,|\phi\rangle\,|\,\mu)\leq\varpi(|0\rangle\,|\,\mu)+\varpi(|\phi\rangle\,|\,\mu).\label{eq:tripartite-asymmetric-bound}
\end{equation}

There are special triples of quantum states $\{|\psi\rangle,|\phi\rangle,|0\rangle\}$
for which the bound Eq.~(\ref{eq:tripartite-asymmetric-bound}) is
necessarily saturated for all $\mu\in\Delta_{|\psi\rangle}$. \emph{Anti-distinguishable
triples}\footnote{Anti-distinguishability was introduced in Ref.~\cite{Caves02} under
the name ``PP-incompatibility'' and was given the more informative
name of anti-distinguishability in Ref.~\cite{Leifer14b}.} have this property. They are triples $\{|\psi\rangle,|\phi\rangle,|0\rangle\}$
which have a quantum measurement with outcomes $E_{\text{\textlnot\ensuremath{\psi}}},E_{\text{\textlnot\ensuremath{\phi}}},E_{\text{\textlnot0}}$
where the probability of obtaining outcome $E_{\text{\textlnot\ensuremath{\psi}}}$
for a system in state $|\psi\rangle$ is zero (and similarly for $E_{\text{\textlnot\ensuremath{\phi}}},E_{\text{\textlnot0}}$).
In other words, there is a measurement which can, with certainty,
identify one state that was definitely not prepared. For this to be
possible, almost no ontic states can be accessible by preparing all three states in the triple, because
any such ontic state wouldn't be able to return any of the outcomes
in the measurement. A full proof of this is provided in Appendix~\ref{sec:app:Anti-Distinguishability-and-Asymmmetric-Overlap}.

It is known that\footnote{This result was proved in Ref.~\cite{Caves02} but Ref.~\cite{Barrett14}
points out and corrects a typographical error in their result (the
original had the second inequality as a strict inequality, which is
incorrect).} triples of states $\{|\psi\rangle,|\phi\rangle,|0\rangle\}$ with
inner products $a=|\langle\psi|\phi\rangle|^{2}$, \textbf{$b=|\langle\psi|0\rangle|^{2}$},
and $c=|\langle\phi|0\rangle|^{2}$ satisfying 
\begin{equation}
a+b+c<1\quad\mathrm{and}\quad(1-a-b-c)^{2}\geq4abc\label{eq:anti-distinguishable-inequalities}
\end{equation}
are necessarily anti-distinguishable by a projective measurement.

Finally, it is useful to consider the effect that unitary transformations
have on asymmetric overlaps. If some transformation $U$ is applied,
via some stochastic map $\gamma\in\Gamma_{U}$, to a system prepared
according to measure $\mu\in\Delta_{|\psi\rangle}$ then the overlap
with some quantum state $|\phi\rangle$ cannot decrease. That is
\begin{equation}
\varpi(U|\phi\rangle\,|\,\mu^{\prime})\geq\varpi(|\phi\rangle\,|\,\mu)\label{eq:overlap-transformation-bound}
\end{equation}
where $\mu^{\prime}\in\Delta_{U|\psi\rangle}$ is the preparation
measure obtained by applying $\gamma$ to $\mu$ as in Eq.~(\ref{eq:Ontological-model-transformation}).
This is because any ontic state accessible by preparing $\mu$ will
be mapped by $\gamma$ onto ontic states accessible by preparing $\mu^{\prime}$
(since preparing $\mu$ then applying $\gamma$ is a preparation of
$\mu^{\prime}$). Similarly, ontic states accessible by preparing
$|\phi\rangle$ map onto states preparable by $U|\phi\rangle$. Thus
ontic states in the overlap of $\mu$ and $|\phi\rangle$ are mapped
by $\gamma$ onto states in the overlap of $\mu^{\prime}$ and $U|\phi\rangle$.
Once again, a full proof is provided in Appendix~\ref{sec:app:Asymmetric-Overlaps-and-Transformations}.

\subsection{Macro-realism for Quantum Systems}

Having laid the groundwork the above definition of macro-realism can
now be applied to quantum systems.

First consider what can count as a ``macroscopically observable''
quantity $Q$. To be observable $Q$ must correspond to some quantum
measurement $M_{Q}$. Therefore, there is some orthonormal basis $\mathcal{B}_{Q}$
so that for each value $q$ of $Q$ the corresponding outcome of $M_{Q}$
is a state in $\mathcal{B}_{Q}$. In order to make sense of the above
definitions $Q$ must also have operational eigenstates for each value
$q$ of $Q$. Fortunately this is straightforward in quantum theory:
every state in $\mathcal{B}_{Q}$ is an operational eigenstate. Moreover,
because the elements of $\mathcal{B}_{Q}$ are orthogonal it follows
that preparations corresponding to different values $q,q^{\prime}$
of $Q$ are therefore distinguishable.

Now consider an ESMR or EMMR model for a quantum system. For any state
$|\psi\rangle\in\mathcal{P}(\mathcal{H})$ and any eigenstate $|0\rangle\in\mathcal{B}_{Q}$
the asymmetric overlap $\varpi(|0\rangle\,|\,\mu)$ for any $\mu\in\Delta_{|\psi\rangle}$
must be maximal. That is,
\begin{equation}
\varpi(|0\rangle\,|\,\mu)=|\langle0|\psi\rangle|^{2},\quad\forall\mu\in\Delta_{|\psi\rangle}.\label{eq:MR-saturated-overlap}
\end{equation}
A full proof of this is provided in Appendix~\ref{sec:app:ESMR,-EMMR,-and-asymmetric-overlap}
and only an outline provided here. Since each state in $\mathcal{B}_{Q}$
is orthogonal to every other, no ontic state (up to measure-zero exceptions)
can be accessible by preparing more than one state in $\mathcal{B}_{Q}$.
Moreover---by definition of ESMR and EMMR---every ontic state must
be accessible by preparing \emph{some} operational eigenstate of $Q$.
Therefore, any ontic state accessible by preparing $\mu$ must also
be accessible by preparing \emph{exactly one} state in $\mathcal{B}_{Q}$
(up to measure-zero sets of exceptions). So the sum $\sum_{i}\varpi(|i\rangle|\mu)=1$
because each overlap is disjoint (up to measure-zero sets of exceptions)
and by Eq.~(\ref{eq:overlap-bound}) each must be maximal, giving
Eq.~(\ref{eq:MR-saturated-overlap}).

Equation~(\ref{eq:MR-saturated-overlap}) is the key consequence
of ESMR and EMMR that leads to the no-go theorem with quantum systems
presented in section~\ref{sec:A-Stronger-No-Go}.

\section{Loopholes in the Leggett-Garg Proof\label{sec:Loopholes-in-the}}

The aim of the LGIs has always been to rule out macro-realist ontologies
for quantum theory when the inequalities are violated. However, in
light of the above precise definition of macro-realism, some loopholes
in the argument can be identified.

The first loophole is that violation of the LGIs cannot rule out SSMR
models of quantum systems. Indeed, no argument that rests on compatibility
with quantum predictions can completely rule out SSMR models since
there exists a well-known SSMR model for quantum systems that reproduces
all quantum predictions: Bohmian mechanics \cite{Bohm52a,Bohm52b,deBroglie27}.

To see that Bohmian mechanics provides an SSMR ontology consider,
for example, the Bohmian description of a single spinless point particle
in three-dimensional space (the argument for more general systems
is analogous). Bohmian mechanics has the ontic state as a pair $\lambda=(\vec{r},|\psi\rangle)\in\mathbb{R}^{3}\times\mathcal{P}(\mathcal{H})$
where $\vec{r}$ is the actual position of the particle and $|\psi\rangle$
is the quantum state (or ``pilot wave''). Note that the quantum
state is part of the ontology here. The ``macroscopically observable
property'' is the position of the particle, $\vec{r}$, and any sharp
measurement of position will reveal the true value of $\vec{r}$ with
certainty. Thus, for any ontic state $\lambda$ there is some value
of the macroscopically observable property (that is, $\vec{r}$) which
is obtained with certainty from any appropriate measurement. Thus,
Bohmian mechanics provides an SSMR ontological model.

The second loophole is that LGI violation is also unable to rule out
ESMR ontological models. This is also demonstrated through a counter-example
in the form of the Kochen-Specker model for the qubit \cite{Kochen67},
which is an ontological model satisfying ESMR\footnote{Strictly speaking, the Kochen-Specker model was not defined with a
post-measurement update rule and so cannot deal with sequences of
measurements (and therefore Leggett-Garg experiments). However, it
is simple to append the obvious update rule ``prepare a new state
corresponding to the measurement outcome'' and this fixes this issue.}. The Kochen-Specker model exactly reproduces quantum predictions
for $d=2$ dimensional Hilbert Spaces. As the LGIs are defined in
$d=2$ the Kochen-Specker model will violate them.

A key question is \emph{why} these counter-examples evade the Leggett-Garg
argument. To derive the LGIs, one needs an additional assumption:
\emph{non-invasive measurability}. 
The Leggett-Garg approach compares the non-invasiveness of the measurement process on operational eigenstates with the invasiveness on preparations that are not operational eigenstates. Their violation shows that these other preparations cannot be expressed as mixtures of operational eigenstates. In neither ESMR nor SSMR can generic preparations be related to mixtures of operational eigenstates, so the Leggett Garg approach generically has loopholes for these types of macro-realism \cite{Maroney14}.
Bohmian mechanics and the Kochen-Specker model are both examples: they contain measurement disturbances that violate the non-invasive measurability assumption, while still satisfying SSMR and ESMR respectively. The crux is that both SSMR and ESMR models can include measurements that don't disturb the distribution over $\Lambda$ if the system is prepared in an operational eigenstate, but still disturb the distribution over $\Lambda$ for systems prepared in other ways.

EMMR, by contrast, requires that all preparations are represented
by statistical mixtures of operational eigenstates. If it can be demonstrated
that operational eigenstates are not disturbed by a given measurement,
then according to EMMR no preparations can be disturbed by that measurement.
It is this feature that prevents EMMR models from violating the LGI
(see Ref.~\cite{Maroney14} for a more extensive discussion of this
point).

Recent experiments \cite{Knee+16,HuffmanMizel16} following Ref. \cite{WildeMizel12}
have sought to address a ``clumsiness loophole'' in Leggett-Garg.
They drop non-invasive measurability as an assumption by incorporating
control experiments to check the disturbance of the measurement on
the operational eigenstates. These approaches follow the Leggett-Garg
argument quite closely and show that the disturbance on some general
preparation cannot be explained in terms of disturbances on a statistical
mixture of operational eigenstates. As a result, they are still only
able to rule out EMMR models. 

So the Leggett-Garg proof, even taking into account the clumsiness
loophole, only rules out EMMR macro-realism and leaves loopholes for
SSMR and ESMR. Moreover, the loophole for SSMR models cannot be fully
plugged by any proof because Bohmian mechanics exists as a counter-example.
Similarly, the loophole for ESMR cannot be fully pluged in $d=2$
dimensions, since the Kochen-Specker model exists as a counter-example.
This leaves a clear question: can the ESMR loophole be closed by another
theorem for any $d>2$? Answering this question needs a different
approach, one which does not make use of the measurement disturbance
assumptions at all.

\section{A Stronger No-Go Theorem\label{sec:A-Stronger-No-Go}}

Using the machinery developed above, it is possible to prove a theorem
that rules out both ESMR and EMMR models for quantum systems with
``macroscopically observable'' properties with $n>3$ values. This
theorem is therefore stronger than the Leggett-Garg proof as it rules
out ESMR models as well as EMMR models.

First, assume there is an ontological model for a quantum system with
$d>3$ dimensions which is ESMR or EMMR for quantity $Q$ with $n=d>3$
values. By applying Eqs.~(\ref{eq:tripartite-asymmetric-bound},\ref{eq:overlap-transformation-bound},\ref{eq:MR-saturated-overlap}) to specially chosen quantum states it is possible to prove a contradiction. 

Let $|0\rangle\in\mathcal{B}_{Q}$ be the eigenstate of some value
$q$ of $Q$ and let $|\psi\rangle\in\mathcal{P}(\mathcal{H})$ be any other state of the system
such that $|\langle0|\psi\rangle|^{2}\in(0,\frac{1}{2})$. Since quantum states in $\mathcal{P}(\mathcal{H})$ are equivalent up to global phase, $\langle 0|\psi\rangle$ can be taken to be a positive real without loss of generality. Now select another orthonormal basis $\mathcal{B}^{\prime}=\{|0\rangle\}\cup\{|i^{\prime}\rangle\}_{i=1}^{d-1}$
for $\mathcal{H}$ such that
\begin{eqnarray}
|\psi\rangle & = & \alpha|0\rangle+\beta|1^{\prime}\rangle+\tau|2^{\prime}\rangle,\\
\alpha & \in & \left(0,\frac{1}{\sqrt{2}}\right),\\
\beta & \eqdef & \sqrt{2}\,\alpha^{2},
\end{eqnarray}
such a basis always exists since, for any $\alpha$, a $\tau\in(0,1)$ exists such that $|\psi\rangle$ is normalised. Define another state with respect to the same basis
\begin{eqnarray}
|\phi\rangle & \eqdef & \delta|0\rangle+\eta|1^{\prime}\rangle+\kappa|3^{\prime}\rangle,\\
\delta & \eqdef & 1-2\alpha^{2},\\
\eta & \eqdef & \sqrt{2}\,\alpha.
\end{eqnarray}

These states have been chosen such that $\langle0|\psi\rangle=\alpha=\langle\phi|\psi\rangle$
and therefore there is some unitary $U$ satisfying $U|0\rangle=|\phi\rangle$
and $U|\psi\rangle=|\psi\rangle$. Moreover, the inner products of
$\{|\psi\rangle,|\phi\rangle,|0\rangle\}$ satisfy Eq.~(\ref{eq:anti-distinguishable-inequalities})
meaning that $\{|\psi\rangle,|\phi\rangle,|0\rangle\}$ is an anti-distinguishable
triple and therefore satisfies Eq.~(\ref{eq:tripartite-asymmetric-bound})
with equality.

Choose any preparation measure $\mu^{\prime}\in\Delta_{|\psi\rangle}$
for $|\psi\rangle$ and any stochastic map $\gamma\in\Gamma_{U}$
for $U$. If $|\psi\rangle$ is prepared according to $\mu^{\prime}$
and then transformed according to $\gamma$ such that $\mu\in\Delta_{|\psi\rangle}$
is the resulting preparation measure (as in Eq.~(\ref{eq:Ontological-model-transformation}))
then 
\begin{eqnarray}
\varpi(|\phi\rangle,|0\rangle\,|\,\mu) & = & \varpi(|\phi\rangle\,|\,\mu)+\varpi(|0\rangle\,|\,\mu)\\
 & \ge & \varpi(|0\rangle\,|\,\mu^{\prime})+\varpi(|0\rangle\,|\,\mu)\\
 & = & 2|\langle0|\psi\rangle|^{2}=2\alpha^{2}
\end{eqnarray}
where the first line follows from anti-distinguishability, the second
from Eq.~(\ref{eq:overlap-transformation-bound}), and the third
from ESMR/EMMR via Eq.~(\ref{eq:MR-saturated-overlap}).

Now consider that $\varpi(|\phi\rangle,|0\rangle\,|\,\mu)\leq\mathbb{P}_{\mathcal{B}^{\prime}}(|0\rangle,|1^{\prime}\rangle\,|\,|\psi\rangle)$.
That is, $\varpi(|\phi\rangle,|0\rangle\,|\,\mu)$ is a lower bound
on the probability that a quantum basis measurement in $\mathcal{B}^{\prime}$
has outcome $|0\rangle$ or $|1^{\prime}\rangle$ for a preparation
of $|\psi\rangle$. To see this, consider that any ontic state preparable
by both $|\psi\rangle$ and $|\phi\rangle$ must return a measurement
outcome compatible with both preparations: thus the outcome must be
either $|0\rangle$ or $|1^{\prime}\rangle$. Similarly any ontic
state preparable by both $|\psi\rangle$ and $|0\rangle$ must return
the outcome $|0\rangle$ in such a measurement. A full proof of this
is provided in Appendix~\ref{sec:app:Tripartite-Asymmetric-Overlap}.

Putting these inequalities together, one finds
\begin{eqnarray}
2\alpha^{2} & \leq & \mathbb{P}_{\mathcal{B}^{\prime}}(|0\rangle,|1^{\prime}\rangle\,|\,|\psi\rangle)=\alpha^{2}+\beta^{2}=\alpha^{2}\left(1+2\alpha^{2}\right)
\end{eqnarray}
implying
\begin{equation}
\alpha\geq\frac{1}{\sqrt{2}}.
\end{equation}
But $\alpha$ was defined to be in the range $(0,\frac{1}{\sqrt{2}})$
and so this is a contradiction. 

The following summarises the assumptions which together imply this
contradiction:
\begin{enumerate}
\item The ontology satisfies ESMR or EMMR.
\item The ``macroscopically observable property'' $Q$ has $n>3$ distinguishable
values (requiring that the quantum system has $d\geq n>3$ dimensions).
\item An eigenstate $|0\rangle$ of $Q$ can be chosen such that some quantum
state $|\psi\rangle$ satisfying $|\langle0|\psi\rangle|\in(0,\frac{1}{\sqrt{2}})$
can be prepared.
\item The quantum transformation $U$ and quantum measurement $\mathcal{B}^{\prime}$
described above can be performed.
\end{enumerate}
Assumptions (i-ii) are about possible underlying ontological models,
while (iii-iv) are implications of standard quantum theory. The conclusion
must therefore be that either ESMR/EMMR ontologies are impossible
for $n>3$ distinguishable values, or that quantum theory is not correct.
Quantum theory is therefore incompatible with ESMR or EMMR macro-realism.

\section{Discussion\label{sec:Discussion}}

The theorem in this paper proves that quantum theory is incompatible
with ESMR and EMMR macro-realist ontologies where the macroscopically
observable property has $n>3$ distinguishable values. This is stronger
than the argument from the Leggett-Garg inequalities, which is only
able to rule out EMMR ontologies. Therefore, only SSMR models are
left as possibilities for macro-realist quantum ontologies.

As noted above, no argument is able to rule out all SSMR ontologies
because Bohmian mechanics is an SSMR theory which reproduces all predictions
of quantum theory. It may be possible, however, to produce a theorem
that rules out some subset of SSMR theories.

For example, Bohmian mechanics is a \emph{$\psi$-ontic} theory \cite{Leifer14b}.
That is, each ontic state $\lambda=(\vec{r},|\psi\rangle)$ can only
be accessed by preparing one quantum state, namely $|\psi\rangle$---there
is no ontic overlap between different quantum states. It may therefore
be possible to prove the incompatibility of quantum theory and all SSMR
ontologies that aren't $\psi$-ontic. This, together with the result
presented here, would essentially say that to be macro-realist you
must have an ontology consisting of the full quantum state plus extra
information. Many would consider this a very strong argument against
macro-realism. For example, such models might reasonably be accused of simply artificially adding macro-realism on top of quantum theory, rather than providing an understanding of quantum theory that respects macro-realism. Of course, those sympathetic to Bohmian mechanics would not be swayed by any such arguments, as Bohmian mechanics is already macro-realist.

Papers on the Leggett-Garg argument, including those addressing the
clumsiness loophole \cite{WildeMizel12,Knee+16,HuffmanMizel16}, have
concentrated on $d=2$ dimensional systems. As a result of closely
following the Leggett-Garg assumptions, they are still unable to rule
out any models outside of EMMR. They certainly could not rule out
all ESMR models, since the Kochen-Specker model satisfies ESMR and
exists in $d=2$. By contrast, the theorem in this paper works when
$d\geq n>3$.

The next stage is the development of an experimental test of the improved
theorem. This requires a detailed, noise-tolerant analysis, as any
experiment is unavoidably subject to non-zero noise. The asymmetric
overlap measure, which is used to characterise the different categories
of macro-realism, is an inherently noise-intolerant quantity. To bring
the theorem into an experimentally testable form therefore requires
a more noise-tolerant alternative and suitable adjusted characterisations
of the different categories of macro-realism. This is possible, but
it is not a simple process. Two different approaches for such noise-tolerant
replacements are currently in development and will be the subject
of future papers \cite{Allen17,Hermens17}.

It is interesting to note that experiments based on this result will
be an entirely new avenue for tests of macro-realism. Experimental
tests based on the Leggett-Garg argument will always have certain
features and difficulties in common (such as the clumsiness loophole
noted in section \ref{sec:Loopholes-in-the}). However, since the
approach of this work is so different in character one can expect
the resulting experiments to be similarly different, hopefully avoiding
many of the difficulties common to Leggett-Garg while requiring challenging
new high-precision tests of quantum theory in $d>2$ Hilbert spaces.

As is common in such foundational works, this paper has considered only the case of finite-dimensional quantum systems. It is hoped that an extension to infinite-dimensional cases should be possible. Due to the fact that quantum states become integrals over bases in the infinite case, a further layer of measure-theoretic complexity would likely be required. Actually developing such an extension therefore remains an interesting open problem.

Finally, one should note that in this paper the ``macro'' quantity
$Q$ was taken to correspond to a measurement of basis $\mathcal{B}_{Q}$
in the quantum case. A more general approach might allow $Q$ to correspond
to a POVM measurement instead. That is, for each value $q$ of $Q$
there would be some POVM element $E_{q}$ and the operational eigenstates
$|\psi\rangle$ of $q$ would be those satisfying $\langle\psi|E_{q}|\psi\rangle=1$.
We are confident that the results presented here can be fairly directly
extended to such a case and this would be another interesting avenue
for further work. Such an extension would likely add significant complexity
to the proofs without changing the fundamental ideas, however.

\begin{acknowledgements}
We would like to thank Jonathan Barrett, Dominic C. Horsman, Matty
Hoban, Ciarán Lee, Chris Timpson, Ronnie Hermens, and Andrew Briggs
for insightful discussions as well as anonymous referees for constructive feedback.\protect \\
JMAA is supported by the Engineering and Physical Sciences Research
Council, the NQIT Quantum Hub, the FQXi Large Grant “Thermodynamic
vs information theoretic entropies in probabilistic theories”, and
the Perimeter Institute for Theoretical Physics. Research at Perimeter
Institute is supported by the Government of Canada through the Department
of Innovation, Science and Economic Development Canada and by the
Province of Ontario through the Ministry of Research, Innovation and
Science.\protect \\
OJEM is supported by the Templeton World Charity Foundation.\protect \\
SG is supported by the Engineering and Physical Sciences Research
Council and Trinity College, Oxford (Williams Scholarship).
\end{acknowledgements}

\bibliography{references}

\newpage
\appendix

\section{A Useful Lemma} \label{sec:app:A-Useful-Lemma}

Here, a lemma is proved that will be useful in the following proofs.
\begin{lem*}
Let $f(\lambda):\Lambda\rightarrow[0,1]$ be any measurable function
from an ontic state space onto the unit interval. Let $\mu$ be any
probability measure over $\Lambda$. If
\begin{equation}
\int_{\Lambda}\mathrm{d}\mu(\lambda)\,f(\lambda)=1\label{eq:lemma-antecedent}
\end{equation}
then it follows that
\begin{equation}
\mu\left(\ker(1-f)\right)=1\label{eq:lemma-conclusion}
\end{equation}
where the kernel of a measurable function $g:\Lambda\rightarrow[0,1]$
is defined
\begin{equation}
\ker g\eqdef\{\lambda^{\prime}\in\Lambda\,:\,g(\lambda^\prime)=0\}\label{eq:kernel-definition}
\end{equation}
and is necessarily a measurable set for measurable function $g$.\end{lem*}
\begin{proof}
Let $\bar{f}(\lambda)\eqdef1-f(\lambda)$, which is a measurable function
from $\Lambda$ to $[0,1]$. As measurable functions, the kernels
of both $f$ and $\bar{f}$ are measurable sets in $\Lambda$.

Equation~(\ref{eq:lemma-antecedent}) implies that
\begin{eqnarray}
1 & = & \int_{\ker\bar{f}}\mathrm{d}\mu(\lambda)\,f(\lambda)+\int_{\Lambda\backslash\ker\bar{f}}\mathrm{d}\mu(\lambda)\,f(\lambda)\\
 & = & \mu(\ker\bar{f})+\int_{\Lambda\backslash\ker\bar{f}}\mathrm{d}\mu(\lambda)\,f(\lambda)\label{eq:lemma-expanded}
\end{eqnarray}
since if $\lambda\in\ker\bar{f}$ then $f(\lambda)=1$. 

Consider the second term in Eq.~(\ref{eq:lemma-expanded}); there
are two possibilities. First, the term could be zero (if, for example,
the only subsets of $\Lambda\backslash\ker\bar{f}$ where $f(\lambda)>0$
are measure-zero subsets according to $\mu$). If it's not zero then,
since for all $\lambda\in\Lambda\backslash\ker\bar{f}$ $f(\lambda)<1$
then $\int_{\Lambda\backslash\ker\bar{f}}\mathrm{d}\mu(\lambda)\,f(\lambda)<\mu(\Lambda\backslash\ker\bar{f})$.
In the latter case this would imply
\begin{eqnarray}
\mu(\ker\bar{f})+\mu(\Lambda\backslash\ker\bar{f}) & > & 1\\
\mu(\Lambda) & > & 1
\end{eqnarray}
which is a contradiction as $\mu(\Lambda)=1$ by definition of a probability
measure. The only option is therefore for this term to vanish, in
which case $1=\mu(\ker\bar{f})=\mu\left(\ker(1-f)\right)$ as desired.
\end{proof}

\section{Bounding Asymmetric Overlaps\label{sec:app:Bounding-Asymmetric-Overlaps}}

The main text quotes a bound, Eq.~(\ref{eq:overlap-bound}), for
the asymmetric overlap together with a sketch of the proof. This will
now be proved fully.

The spirit of the proof is that if $\lambda$ can be obtained by preparing
some $\nu\in\Delta_{|\phi\rangle}$ then it should also return the
outcome $|\phi\rangle$ with certainty in any measurement where that
is an option (there are exceptions which make the proof more difficult
than this). Thus, the probability of preparing $\mu$ and getting
a $\lambda$ which then returns the outcome $|\phi\rangle$ in a measurement
is at least the probability of preparing $\mu$ and getting a $\lambda$
which is accessible from some $\nu\in\Delta_{|\phi\rangle}$, which
is a paraphrase of the desired result.

To prove Eq.~(\ref{eq:overlap-bound}), consider preparing $|\phi\rangle$
according to $\nu\in\Delta_{|\phi\rangle}$ and then performing some
quantum measurement $M_{\phi}\ni|\phi\rangle$. For the ontological
model to reproduce quantum probabilities, as in Eq.~(\ref{eq:ontological-model-quantum-probability}),
it is therefore required that
\begin{equation}
\int_{\Lambda}\mathrm{d}\nu(\lambda)\mathbb{P}_{M_{\phi}}(|\phi\rangle\,|\,\lambda)=1.\label{eq:self-measure-ontic-probability}
\end{equation}

It is tempting to conclude from this that for every $\lambda$ in
the support of $\nu$ then $\mathbb{P}_{M_{\phi}}(|\phi\rangle\,|\,\lambda)=1$.
However, this is not generally possible because it is not generally
possible to define a support for the measure $\nu$. Therefore, a
slightly different proof is required.

For notational convenience, define $g(\lambda)\eqdef\mathbb{P}_{M_{\phi}}(|\phi\rangle\,|\,\lambda)$
as a measurable function from $\Lambda$ to $[0,1]$. It follows from
Eq.~(\ref{eq:self-measure-ontic-probability}) and the lemma of Appendix~\ref{sec:app:A-Useful-Lemma}
that $\nu(\ker(1-g))=1$. That is, $g(\lambda)=1$ almost everywhere
according to $\nu$.

Now consider preparing $|\psi\rangle$ according to $\mu$ and then
measuring with the same $M_{\phi}\ni|\phi\rangle$, by Eq.~(\ref{eq:Ontological-model-probability})
and the Born rule
\begin{equation}
|\langle\phi|\psi\rangle|^{2}=\int_{\Lambda}\mathrm{d}\mu(\lambda)\mathbb{P}_{M_{\phi}}(|\phi\rangle\,|\,\lambda).
\end{equation}
Restricting the integral to $\ker(1-g)$ where, for every $\lambda$,
$g(\lambda)=1$ one finds
\begin{eqnarray}
|\langle\phi|\psi\rangle|^{2} & \ge & \int_{\ker(1-g)}\mathrm{d}\mu(\lambda)\,g(\lambda)\\
 & = & \mu\left(\ker(1-g)\right).
\end{eqnarray}
Thus, $|\langle\phi|\psi\rangle|^{2}$ bounds $\mu(\Omega)$ from
above for some $\Omega\subseteq\Lambda$ for which $\nu(\Omega)=1$.
It must therefore bound the smallest such $\mu(\Omega)$ from above:
\begin{equation}
|\langle\phi|\psi\rangle|^{2}\geq\mu\left(\ker(1-g)\right)\geq\inf_{\Omega\subseteq\Lambda:\nu(\Omega)=1}\mu(\Omega)\eqdef\varpi(\nu\,|\,\mu)
\end{equation}
which proves Eq.~(\ref{eq:overlap-bound}).

\section{Asymmetric Overlaps and Transformations} \label{sec:app:Asymmetric-Overlaps-and-Transformations}

A similar manoeuvre to Appendix~\ref{sec:app:Bounding-Asymmetric-Overlaps}
can be used to prove Eq.~(\ref{eq:overlap-transformation-bound}).
The gist of this proof is given in the text while the full proof is
below.

It suffices to prove that for any measurable $\Omega^{\prime}\subseteq\Lambda$
such that $\nu^{\prime}(\Omega^{\prime})=1,\forall\nu^{\prime}\in\Delta_{U|\phi\rangle}$
then there exists some measurable $\Omega\subseteq\Lambda$ such that
$\nu(\Omega)=1,\forall\nu\in\Delta_{|\phi\rangle}$ and $\mu^{\prime}(\Omega^{\prime})\geq\mu(\Omega)$.
Since $\mu^{\prime}$ is the measure obtained by following $\mu$
with $\gamma$ then by Eq.~(\ref{eq:Ontological-model-transformation})
\begin{equation}
\mu^{\prime}(\Omega^{\prime})\eqdef\int_{\Lambda}\mathrm{d}\mu(\lambda)\gamma(\Omega^{\prime}|\lambda)\label{eq:proof-transform-definition}
\end{equation}
where $\Omega^{\prime}$ is any such subset.

For any $\nu\in\Delta_{|\phi\rangle}$ there is some $\nu^{\prime}\in\Delta_{U|\phi\rangle}$
such that $\nu^{\prime}$ is obtained by performing $\nu$ followed
by $\gamma$. Therefore, by Eq.~(\ref{eq:Ontological-model-transformation})
\begin{equation}
1=\nu^{\prime}(\Omega^{\prime})=\int_{\Lambda}\mathrm{d}\nu(\lambda)\gamma(\Omega^{\prime}|\lambda)\label{eq:self-transform-probability}
\end{equation}
Viewing $f(\lambda)\eqdef\gamma(\Omega^{\prime}|\lambda)$ as a measurable
function from $\Lambda$ to $[0,1]$ and using the lemma of Appendix~\ref{sec:app:A-Useful-Lemma},
Eq.~(\ref{eq:self-transform-probability}) implies $\nu(\ker(1-f))=1$.
Note that this holds for every $\nu\in\Delta_{|\phi\rangle}$.

Therefore, let $\Omega=\ker(1-f)\subseteq\Lambda$ be a subset for
which every $\nu(\Omega)=1,\forall\nu\in\Delta_{|\phi\rangle}$ and
return to Eq.~(\ref{eq:proof-transform-definition}). One finds
\begin{eqnarray}
\mu^{\prime}(\Omega^{\prime}) & \geq & \int_{\ker(1-f)}\mathrm{d}\mu(\lambda)\,f(\lambda)\\
 & = & \mu(\Omega)
\end{eqnarray}
which proves the desired statement and thereby Eq.~(\ref{eq:overlap-transformation-bound}).

\section{Anti-Distinguishability and Asymmetric Overlap} \label{sec:app:Anti-Distinguishability-and-Asymmmetric-Overlap}

It is claimed (and roughly motivated) in the text that if $\{|\psi\rangle,|\phi\rangle,|0\rangle\}$
is an anti-distinguishable triple then Eq.~(\ref{eq:tripartite-asymmetric-bound})
must hold with equality. This will now be proved fully.

Recall that $\{|\psi\rangle,|\phi\rangle,|0\rangle\}$ is an anti-distinguishable
triple if and only if there is some quantum measurement $M$ with
three outcomes $E_{\text{\textlnot}\psi},E_{\text{\textlnot}\phi},E_{\text{\textlnot}0}$
such that the outcome of getting $E_{\text{\textlnot}\psi}$ from
a system prepared in state $|\psi\rangle$ is zero and similarly for
the other state/outcome pairs. By Eq.~(\ref{eq:ontological-model-quantum-probability})
it therefore follows that for all $\mu\in\Delta_{|\psi\rangle},\nu\in\Delta_{|\phi\rangle},\chi\in\Delta_{|0\rangle}$
\begin{eqnarray}
\int_{\Lambda}\mathrm{d}\mu(\lambda)\mathbb{P}_{M}(E_{\text{\textlnot}\psi}|\,\lambda) & = & 0\label{eq:antidistinguishable-measurement-1}\\
\int_{\Lambda}\mathrm{d}\nu(\lambda)\mathbb{P}_{M}(E_{\text{\textlnot}\phi}|\,\lambda) & = & 0\label{eq:antidistinguishable-measurement-2}\\
\int_{\Lambda}\mathrm{d}\chi(\lambda)\mathbb{P}_{M}(E_{\text{\textlnot}0}|\,\lambda) & = & 0.\label{eq:antidistinguishable-measurement-3}
\end{eqnarray}

To prove that Eq.~(\ref{eq:tripartite-asymmetric-bound}) holds with
equality it suffices to show that given any measurable $\Omega\subseteq\Lambda$
for which $\nu(\Omega)=\chi(\Omega)=1$ for all $\nu\in\Delta_{|\phi\rangle},\chi\in\Delta_{|0\rangle}$
there exists some measurable $\Omega^{\prime},\Omega^{\prime\prime}\subseteq\Lambda$
for which $\nu(\Omega^{\prime})=\chi(\Omega^{\prime\prime})=1$ for
all $\nu\in\Delta_{|\phi\rangle},\chi\in\Delta_{|0\rangle}$ for which
\begin{equation}
\mu(\Omega)\geq\mu(\Omega^{\prime})+\mu(\Omega^{\prime\prime}).\label{eq:antidistinguishability-equality-to-prove}
\end{equation}
This, together with Eq.~(\ref{eq:tripartite-asymmetric-bound}) itself,
would prove the desired result since the right-hand side bounds $\varpi(|\phi\rangle|\mu)+\varpi(|0\rangle|\mu)$
from above.

To prove that Eq.~(\ref{eq:antidistinguishability-equality-to-prove})
holds define the following measurable functions from $\Lambda$ to
$[0,1]$
\begin{eqnarray}
g_{\psi}(\lambda) & \eqdef & \mathbb{P}_{M}(E_{\text{\textlnot}\phi}|\,\lambda)+\mathbb{P}_{M}(E_{\text{\textlnot}0}|\,\lambda)\\
g_{\phi}(\lambda) & \eqdef & \mathbb{P}_{M}(E_{\text{\textlnot}\psi}|\,\lambda)+\mathbb{P}_{M}(E_{\text{\textlnot}0}|\,\lambda)\\
g_{0}(\lambda) & \eqdef & \mathbb{P}_{M}(E_{\text{\textlnot}\psi}|\,\lambda)+\mathbb{P}_{M}(E_{\text{\textlnot}\phi}|\,\lambda).
\end{eqnarray}
Using the fact that, for any $\lambda\in\Lambda$, the sum of probabilities
of outcomes for any measurement must be unity it follows that $\mathbb{P}_{M}(E_{\text{\textlnot}\psi}|\lambda)=1-g_{\psi}(\lambda)$
and similarly for $|\phi\rangle$ and $|0\rangle$. Therefore Eqs.~(\ref{eq:antidistinguishable-measurement-1},\ref{eq:antidistinguishable-measurement-2},\ref{eq:antidistinguishable-measurement-3})
are equivalent to 
\begin{eqnarray}
\int_{\Lambda}\mathrm{d}\mu(\lambda)\,g_{\psi}(\lambda) & = & 1\\
\int_{\Lambda}\mathrm{d}\nu(\lambda)\,g_{\phi}(\lambda) & = & 1\\
\int_{\Lambda}\mathrm{d}\chi(\lambda)\,g_{0}(\lambda) & = & 1.
\end{eqnarray}
By the lemma in Appendix~\ref{sec:app:A-Useful-Lemma} it immediately follows
that $\nu(\ker(1-g_{\phi}))=\chi(\ker(1-g_{0}))=1$ where, recall,
$\nu$ and $\chi$ are arbitrary measures from $\Delta_{|\phi\rangle}$
and $\Delta_{|0\rangle}$ respectively.

With these definitions, consider $\mu(\Omega)$ for any measurable
$\Omega\subseteq\Lambda$ for which $\nu(\Omega)=\chi(\Omega)=1$
for all $\nu\in\Delta_{|\phi\rangle},\chi\in\Delta_{|0\rangle}$
\begin{eqnarray}
\mu(\Omega) & = & \int_{\Omega}\mathrm{d}\mu(\lambda)\\
 & = & \int_{\Omega}\mathrm{d}\mu(\lambda)\Bigl(\mathbb{P}_{M}(E_{\text{\textlnot}\psi}|\lambda)+g_{\psi}(\lambda)\Bigr)\\
 & = & \int_{\Omega}\mathrm{d}\mu(\lambda)\,g_{\psi}(\lambda)+\int_{\Omega}\mathrm{d}\mu(\lambda)\mathbb{P}_{M}(E_{\text{\textlnot}\psi}|\lambda)
\end{eqnarray}
By the definitions of $g_\psi$, $g_\phi$, and $g_0$ this equals
\begin{eqnarray}
\mu(\Omega) & = & \int_{\Omega}\mathrm{d}\mu(\lambda)\,g_{\phi}(\lambda)+\int_{\Omega}\mathrm{d}\mu(\lambda)\,g_{0}(\lambda)\nonumber \\
 &  & -\int_{\Omega}\mathrm{d}\mu(\lambda)\,\mathbb{P}_{M}(E_{\text{\textlnot}\psi}|\lambda)\\
 & = & \int_{\Omega}\mathrm{d}\mu(\lambda)\,g_{\phi}(\lambda)+\int_{\Omega}\mathrm{d}\mu(\lambda)\,g_{0}(\lambda)
\end{eqnarray}
where the last term vanished because $0\leq\int_{\Omega}\mathrm{d}\mu(\lambda)\mathbb{P}_{M}(E_{\text{\textlnot}\psi}|\lambda)\leq\int_{\Lambda}\mathrm{d}\mu(\lambda)\mathbb{P}_{M}(E_{\text{\textlnot}\psi}|\lambda)=0$
by Eq.~(\ref{eq:antidistinguishable-measurement-1}). Continuing
\begin{eqnarray}
\mu(\Omega) & \geq & \int_{\Omega\cap\ker(1-g_{\phi})}\mathrm{d}\mu(\lambda)\,g_{\phi}(\lambda)\nonumber \\
 &  & +\int_{\Omega\cap\ker(1-g_{0})}\mathrm{d}\mu(\lambda)\,g_{0}(\lambda)\\
 & = & \mu\left(\Omega\cap\ker(1-g_{\phi})\right)+\mu\left(\Omega\cap\ker(1-g_{0})\right)
\end{eqnarray}
by restricting the domains of integration and recalling that $\forall\lambda\in\ker(1-g_{\phi})$,
$g_{\phi}(\lambda)=1$ (and similarly for $g_{0}$). Note that as
both $\Omega$ and $\ker(1-g_{\phi})$ are measure-one according to
any $\nu\in\Delta_{|\phi\rangle}$, it follows that their intersection
also satisfies $\nu\left(\Omega\cap\ker(1-g_{\phi})\right)=1$. Similarly,
$\chi\left(\Omega\cap\ker(1-g_{0})\right)=1$. Thus what has been
proved is that given any measurable $\Omega$ such that $\nu(\Omega)=\chi(\Omega)=1$
for all $\nu\in\Delta_{|\phi\rangle},\chi\in\Delta_{|0\rangle}$ there
exist measurable sets $\Omega^{\prime}\eqdef\Omega\cap\ker(1-g_{\phi})$
and $\Omega^{\prime\prime}\eqdef\Omega\cap\ker(1-g_{0})$ satisfying
$\nu(\Omega^{\prime})=\chi(\Omega^{\prime\prime})=1$ for all $\nu\in\Delta_{|\phi\rangle},\chi\in\Delta_{|0\rangle}$
for which
\begin{equation}
\mu(\Omega)\geq\mu(\Omega^{\prime})+\mu(\Omega^{\prime\prime}).
\end{equation}
This is exactly what was to be proved.

\section{ESMR, EMMR, and Asymmetric Overlap} \label{sec:app:ESMR,-EMMR,-and-asymmetric-overlap}

Here it is proved that for any state $|\psi\rangle\in\mathcal{P}(\mathcal{H})$
and any eigenstate $|0\rangle\in\mathcal{B}_{Q}$ for an ESMR/EMMR
quantity $Q$ then Eq.~(\ref{eq:MR-saturated-overlap}) must hold.
This follows because in both ESMR and EMMR ontological models every
ontic state that can be prepared can be obtained by preparing some
operational eigenstate of $Q$.

This is easiest to state precisely if the asymmetric overlap is extended
further to include arbitrary $n$-partite sets of states. That is,
if $\mathcal{S}$ is a set of $n-1$ quantum states and $\mu\in\Delta_{|\psi\rangle}$
is a preparation measure for quantum state $|\psi\rangle$ then the
$n$-partite asymmetric overlap $\varpi(\mathcal{S}|\mu)$ is the
probability that the ontic state obtained by sampling $\mu$ could
have been obtained by preparing some state in $\mathcal{S}$. Mathematically,
\begin{equation} \label{eq:multipartite-asymmetric}
\varpi(\mathcal{S}\,|\,\mu)\eqdef\inf\left\{ \mu(\Omega)\;:\;\Omega\subseteq\Lambda,\forall|\phi\rangle\in\mathcal{S},\forall\nu\in\Delta_{|\phi\rangle},\nu(\Omega)=1\right\} .
\end{equation}
As with the tripartite overlap, Boole's inequality gives the bound
\begin{equation}
\varpi(\mathcal{S}\,|\,\mu)\leq\sum_{|\phi\rangle\in\mathcal{S}}\varpi(|\phi\rangle\,|\,\mu).\label{eq:multipartite-asymmetric-booles}
\end{equation}

From Eq.~(\ref{eq:multipartite-asymmetric}), it follows immediately that for any ESMR or EMMR ontological model,
every preparation $\mu\in\Delta_{|\phi\rangle}$ of every state $|\psi\rangle\in\mathcal{P}(\mathcal{H})$
\begin{equation}
\varpi(\mathcal{B}_{Q}\,|\,\mu)=1,\label{eq:mutlipartite-MR-unity}
\end{equation}
\emph{i.e. }the probability of obtaining a $\lambda\in\Lambda$ accessible
from some $|0\rangle\in\mathcal{B}_{Q}$ is unity. Let $\Omega\subseteq\Lambda$
be any measurable subset for which $\nu(\Omega)=1,\forall\nu\in\Delta_{|0\rangle},\forall|0\rangle\in\mathcal{B}_{Q}$.
Let $M_{Q}$ be any measurement for the basis $\mathcal{B}_{Q}$ and,
for each $|0\rangle\in\mathcal{B}_{Q}$, let $f_{|0\rangle}(\lambda)\eqdef\mathbb{P}_{M_{Q}}(|0\rangle\,|\,\lambda)$
be a measurable function from $\Lambda$ to $[0,1]$. By the lemma
of Appendix~\ref{sec:app:A-Useful-Lemma} it follows that for every $|0\rangle\in\mathcal{B}_{Q}$
and any $\nu\in\Delta_{|0\rangle}$
\begin{eqnarray}
 & 1 & =\int_{\Lambda}\mathrm{d}\nu(\lambda)\,f_{|0\rangle}(\lambda)\\
\Rightarrow & 1 & =\nu(\ker(1-f_{|0\rangle})).
\end{eqnarray}
Similarly, for any other $|1\rangle\in\mathcal{B}_{Q}$ and $\chi\in\Delta_{|1\rangle}$
($|1\rangle\neq|0\rangle$) then 
\begin{eqnarray}
 & 0 & =\int_{\Lambda}\mathrm{d}\chi(\lambda)\,f_{|0\rangle}(\lambda)\\
\Rightarrow & 1 & =\int_{\Lambda}\mathrm{d}\chi(\lambda)\,(1-f_{|0\rangle}(\lambda))\\
\Rightarrow & 1 & =\chi(\ker f_{|0\rangle}).
\end{eqnarray}

Using these definitions, for any $|0\rangle\in\mathcal{B}_{Q}$ 
\begin{equation}
\mu(\Omega)=\mu\left(\Omega\cap\ker(1-f_{|0\rangle})\right)+\mu(\Omega\backslash\ker(1-f_{|0\rangle})).
\end{equation}
Since , $\nu(\ker(1-f_{|0\rangle}))=1,\forall\nu\in\Delta_{|0\rangle}$
then $\mu(\Omega\cap\ker(1-f_{|0\rangle}))$ upper bounds $\varpi(|0\rangle\,|\,\mu)$.
Moreover, since $\ker f_{|0\rangle}\subseteq\Lambda\backslash\ker(1-f_{|0\rangle})$
then for any $\chi\in\Delta_{|1\rangle},|1\rangle\in\mathcal{B}_{Q}\backslash\{|0\rangle\}$
it follows that $\chi(\Omega\backslash\ker(1-f_{|0\rangle}))=1$.
Thus this process can be continued taking $\Omega^{\prime}=\Omega\backslash\ker(1-f_{|0\rangle})$
as a set for which $\chi(\Omega^{\prime})=1,\forall\chi\in\Delta_{|1\rangle},\forall|1\rangle\in\mathcal{B}_{Q}\backslash\{|0\rangle\}$
then for some $|1\rangle\in\mathcal{B}_{Q}\backslash\{|0\rangle\}$
\begin{equation}
\mu(\Omega^{\prime})=\mu\left(\Omega^{\prime}\cap\ker(1-f_{|1\rangle})\right)+\mu(\Omega^{\prime}\backslash\ker(1-f_{|1\rangle}))
\end{equation}
where $\mu(\Omega^{\prime}\cap\ker(1-f_{|1\rangle}))$ upper bounds
$\varpi(|1\rangle\,|\,\mu)$.

By induction therefore, one finds,
\begin{equation}
\mu(\Omega)\geq\sum_{|0\rangle\in\mathcal{B}_{Q}}\varpi(|0\rangle\,|\,\mu).
\end{equation}
Since this holds for any such $\Omega$ then
\begin{equation}
\varpi(\mathcal{B}_{Q}\,|\,\mu)\geq\sum_{|0\rangle\in\mathcal{B}_{Q}}\varpi(|0\rangle\,|\,\mu)
\end{equation}
which, combined with Eqs.~(\ref{eq:multipartite-asymmetric-booles},\ref{eq:mutlipartite-MR-unity})
gives
\begin{equation}
\varpi(\mathcal{B}_{Q}\,|\,\mu)=\sum_{|0\rangle\in\mathcal{B}_{Q}}\varpi(|0\rangle\,|\,\mu)=1.\label{eq:basis-multipartite-sum-unity}
\end{equation}

Finally, noting that $\varpi(|0\rangle\,|\,\mu)\leq|\langle0|\psi\rangle|^{2}$
and $\sum_{|0\rangle\in\mathcal{B}_{Q}}|\langle0|\psi\rangle|^{2}=1$
as $\mathcal{B}_{Q}$ is a basis and $|\psi\rangle$ is normalised,
then Eq.~(\ref{eq:basis-multipartite-sum-unity}) implies
\begin{equation}
\varpi(|0\rangle\,|\,\mu)=|\langle0|\psi\rangle|^{2},\quad\forall|0\rangle\in\mathcal{B}_{Q}
\end{equation}
which is exactly Eq.~(\ref{eq:MR-saturated-overlap}).

\section{Tripartite Asymmetric Overlap and Quantum Measurements} \label{sec:app:Tripartite-Asymmetric-Overlap}

In the proof of the main theorem, it is claimed that 
\begin{equation}
\varpi(|\phi\rangle,|0\rangle\,|\,\mu)\leq\mathbb{P}_{\mathcal{B}^{\prime}}(|0\rangle,|1^{\prime}\rangle\,|\,|\psi\rangle),\label{eq:proof-tirpartite-quantum-overlap}
\end{equation}
where the right hand side is the quantum probability of the corresponding
quantum measurement and the states and measurements are as defined
in the text. This will now be proved.

The gist of the proof is that if $\lambda\in\Lambda$ is accessible
by preparing two quantum states then it may only return measurement
results that are compatible with both preparations. So if $\lambda$
is accessible by preparing both $|\psi\rangle$ and $|0\rangle$ then
it may only return $|0\rangle$ in a measurement in the $\mathcal{B}^{\prime}$
basis and similarly if $\lambda$ is accessible by preparing both
$|\psi\rangle$ and $|\phi\rangle$ then it may only return either
$|0\rangle$ or $|1^{\prime}\rangle$ for a similar measurement. Thus,
if one of these $\lambda$'s is obtained in a preparation of $|\psi\rangle$
then the measurement result is necessarily either $|0\rangle$ or
$|1^{\prime}\rangle$. This will now be fully fleshed out.

Consider a measurement $M_{\mathcal{B}^{\prime}}$ of the basis $\mathcal{B}^{\prime}$
and let $f(\lambda)\eqdef\mathbb{P}_{\mathcal{B}^{\prime}}(|0\rangle\,|\,\lambda)+\mathbb{P}_{\mathcal{B}^{\prime}}(|1\rangle\,|\,\lambda)$
and $g_{3}(\lambda)\eqdef\mathbb{P}_{\mathcal{B}^{\prime}}(|3^{\prime}\rangle\,|\,\lambda)$
be measurable functions from $\Lambda$ to $[0,1]$. Note that by
definition of $|\phi\rangle$ and $|0\rangle$the following quantum
probabilities are known for any $\nu\in\Delta_{|\phi\rangle}$ and
$\chi\in\Delta_{|0\rangle}$
\begin{eqnarray}
\int_{\Lambda}\mathrm{d}\nu(\lambda)\left(f(\lambda)+g_{3}(\lambda)\right) & = & 1\\
\int_{\Lambda}\mathrm{d}\chi(\lambda)f(\lambda) & = & 1.
\end{eqnarray}
Therefore by the lemma of Appendix~\ref{sec:app:A-Useful-Lemma} it follows
that $\nu(\ker(1-f-g_{3}))=\chi(\ker(1-f))=1$. Since $\ker(1-f)\subseteq\ker(1-f-g_{3})$
it follows that $\Omega\eqdef\ker(1-f-g_{3})$ is a measurable subset
of $\Lambda$ for which $\nu(\Omega)=\chi(\Omega)=1,\,\forall\nu\in\Delta_{|\phi\rangle},\chi\in\Delta_{|0\rangle}$.

The desired quantum probability is given by
\begin{equation}
\mathbb{P}_{\mathcal{B}^{\prime}}(|0\rangle,|1^{\prime}\rangle\,|\,|\psi\rangle)=\int_{\Lambda}\mathrm{d}\mu(\lambda)\,f(\lambda).
\end{equation}
Since the quantum probability $\int_{\Lambda}\mathrm{d}\mu(\lambda)g_{3}(\lambda)=0$
is also known (by definition of $|\psi\rangle$) it follows that 
\begin{eqnarray}
\mathbb{P}_{\mathcal{B}^{\prime}}(|0\rangle,|1^{\prime}\rangle\,|\,|\psi\rangle) & = & \int_{\Lambda}\mathrm{d}\mu(\lambda)\left(f(\lambda)+g_{3}(\lambda)\right)\\
 & \geq & \int_{\ker(1-f-g_{3})}\mathrm{d}\mu(\lambda)\left(f(\lambda)+g_{3}(\lambda)\right)\\
 & = & \mu\left(\ker(1-f-g_{3})\right)=\mu(\Omega).
\end{eqnarray}
Since, by definition $\mu(\Omega)$ upper bounds $\varpi(|\phi\rangle,|0\rangle\,|\,\mu)$
this proves Eq.~(\ref{eq:proof-tirpartite-quantum-overlap}).
\end{document}